\documentclass[11pt]{amsart}

\usepackage{amsmath,  amssymb, amscd}

\usepackage{fullpage}

\theoremstyle{plain}

\newtheorem{thm}{Theorem}
\newtheorem{theorem}{Theorem}
\newtheorem*{thm*}{Theorem}

\newtheorem*{cor*}{Corollary}

\newtheorem*{conj*}{Conjecture}
\newtheorem*{lemma*}{Lemma}

\newtheorem{lemma}[thm]{Lemma}
\newtheorem*{prop*}{Proposition}

\newtheorem{proposition}[thm]{Proposition}
\theoremstyle{definition}

\newtheorem*{defn*}{Definition}

\newtheorem*{rems*}{Remarks}
\newtheorem*{proof*}{Proof}
\newtheorem{prel*}{Preliminaries}
\newtheorem{examples*}{Examples}

\newcommand\coF{{}^{\mathcal{C}}\kern-2pt\Lambda}

\newcommand\cFTs{{}^{\Phi}\overline{T}\kern-1pt{}^*}

\newcommand\dbar{\overline{\pa}}

\newcommand\cG{\mathcal{G}}

\newcommand\cL{\mathcal{L}}
\newcommand\cO{\mathcal{O}}

\newcommand\CC{\mathbb C}

\newcommand\PP{\mathbb P}

\newcommand\RR{\mathbb R}

\newcommand\ZZ{\mathbb Z}

\newcommand\cC{\mathcal C}

\newcommand\cFNs{{}^{\Phi}\overline N\kern-1pt{}^*}

\newcommand\pa{\partial}

\renewcommand\Re{\operatorname{Re}}

\begin{document}

\title[Holomorphic determinant line bundles]
{Holomorphic Quillen determinant line bundles \\on integral compact K\"ahler manifolds}

\author{Rukmini Dey}
\address{Department of Mathematics,
Harishchandra Research Institute,
Allahabad, India}
\email{rkmn@hri.res.in}

\author{Varghese Mathai}
\address{Department of Pure Mathematics, University of Adelaide,
Adelaide 5005, Australia}
\email{mathai.varghese@adelaide.edu.au}

\dedicatory{Dedicated to the memory of Daniel G. Quillen}

\begin{abstract}

We show that any compact Kahler manifold with integral Kahler form, parametrizes a natural holomorphic family of Cauchy-Riemann operators on the Riemann sphere such that the Quillen determinant line bundle of this family is isomorphic to
a sufficiently high tensor power of the holomorphic line bundle determined by the integral Kahler form. We also establish a symplectic version of the result. We conjecture that an equivariant version of our result is true.
\end{abstract}

\thanks{{\em Acknowledgments.} The first author thanks the SYM project grant 
for support and the second author thanks the Australian Research Council
   for support and the Harishchandra Research Institute for hospitality during the period 
when this work was completed. Both authors thank M.S. Narasimhan for pointing out the
error in the statement of Kodaira's embedding theorem in the previous version of the paper. 
We also thank the referees for suggestions leading to improvements in the paper.The first author would like to thank Professor Leon Takhtajan for introducing 
her to the Quillen determinant bundle during her thesis work. V.M. was fortunate to have been a Ph.D. student
of Daniel Quillen at MIT}
  

\maketitle


\section{Introduction}

In geometric quantization, given a type $(1,1)$ integral form $\xi$ on a compact K\"ahler manifold, there is a holomorphic line bundle $\cL$ on the manifold with connection and curvature equal to $\xi$.  For positive integral $(1,1)$ forms, we prove 
the following partial refinement.

\begin{theorem}\label{thm:main}
Any compact K\"ahler manifold $M$ with integral K\"ahler form $\omega$, parametrizes a natural holomorphic family of 
Cauchy-Riemann operators $\{\dbar_z: z\in M\}$
on $\CC P^1$ such that the Quillen determinant line bundle $\det(\dbar) \cong \cL^{\otimes k}$ as holomorphic line bundles, where $\cL$ is the holomorphic 
line bundle determined by $\omega$, for some sufficiently large $k$. 
\end{theorem}

Here we recall that in the seminal paper \cite{Quillen}, Quillen constructed a holomorphic line bundle $\mathfrak L$, termed the determinant line bundle,  
associated to any holomorphic family of Cauchy-Riemann operators over a compact Riemann surface. Quillen also constructed a hermitian metric on $\mathfrak L$, and calculated its curvature. His construction is referred to in the statement of  Theorem \ref{thm:main}.

The strategy of the proof of Theorem \ref{thm:main} goes briefly as follows. We first establish the theorem for complex projective spaces $\CC P^N$ for all
integers $N$ and for $k=1$. This is achieved by viewing $\CC P^N$ as the moduli space of $N$-vortices on $\CC P^1$. 
The position of each $N$-vortex $z$ determines a Cauchy-Riemann operator $\dbar_z$ on $\CC P^1$, and the 
Quillen determinant line bundle,
denoted $\det(\dbar)$, is isomorphic as a holomorphic line bundle to the hyperplane bundle $\cO(1)$ on $\CC P^N$, see \S \ref{sect:vortex}.
The next step is to apply the Kodaira embedding theorem to establish the theorem in general, see \S \ref{sect:kodaira}. 
In \S \ref{sect:gromov}, we prove 
a symplectic analogue of our theorem, established using Gromov's embedding theorem.
In  \S \ref{sect:conj}, we conjecture that an equivariant version of our result is true.

Our use of the vortex moduli space seems to be novel: the moduli space is usually a subject of investigations for its own sake, rather than a tool for applications. 
As an application of our theorem, recall that one can relate the holomorphic
sections of the hyperplane bundle $\cO(1)$ on $\CC P^k$ to coherent states in physics, cf. ~\cite{R1,S1}.
Our use of the Kodaira embedding theorem shows that the same is true for holomorphic sections of $\cL^{\otimes k}$
for some sufficiently large $k$, and by our theorem, the same is also true for $\det(\dbar) $, which is a topic that we plan to explore. 

\section{Moduli space of N-vortices on a sphere and determinant line bundle}\label{sect:vortex}

\subsection{Vortex equations}

The vortex equations are as follows. Let $\Sigma$ be a compact Riemann 
surface (in our case the sphere) and let $\omega = \frac{1}{2}h^2 dz \wedge d \bar{z}$ be the purely imaginary 
volume form on it, (i.e. $h$ is real). 
Let  $A$ be a unitary connection on a principal $U(1)$ bundle $P$
 i.e. $A$ is a purely imaginary valued one form i.e. $A = A^{(1,0)} + A^{(0,1)}$ such that $A^{(1,0)} = -\overline{A^{(0,1)}}$.   
Let $L$ be a complex line bundle associated to $P$ by the defining 
representation. We denote by the same $A$ the unitary connection on $L$.
  Let $\Psi$ be a section of $L$, i.e. 
$\Psi \in  \Gamma(\Sigma,L)$ and $\bar{\Psi}$ be a section of its dual, $\bar{L}$. 
There is a Hermitian metric $H$ on $L$, i.e. the inner product $<\Psi_1, \Psi_2>_H = \Psi_1 H \bar{\Psi}_2$ is a smooth function on $\Sigma$. (Here $H$ is real). 

The pair $(A, \Psi)$ will be said to satisfy the vortex equations if

\hspace{1in} $(1)$ $ \rm{\;\;\;\;\;}$ $ F(A) = \frac{1}{2} (1-|\Psi|^2_H) \omega,$

\hspace{1in} $(2)$ $\rm{\;\;\;\;\;}$ $\bar{\partial}_A \Psi = 0,$\\
where $F(A)$ is the curvature of the connection $A$ and $d_A = \partial_A + \bar{\partial}_A $ is the decomposition of the covariant derivative operator
into $(1,0)$ and $(0,1)$ pieces.  
Let ${\mathcal S}$ be the space of solutions to $(1)$ and $(2)$.
 There is a gauge group $G$ acting on the space of $(A, \Psi)$ which leaves the equations invariant. We take the group $G$ to be abelian and locally it looks like ${\rm Maps} (\Sigma, U(1)).$ If $g$ is an $U(1)$ gauge transformation then 
$(A_1, \Psi_1)$ and $(A_2, \Psi_2)$ are gauge equivalent if 
$A_2 = g^{-1}dg + A_1 $ and $\Psi_2 = g^{-1} \Psi_1$. 
  Taking the quotient by the gauge group of ${\mathcal S}$ gives  the moduli 
space of solutions to these 
equations and is denoted by ${\mathcal M}$.

$\displaystyle \int_{\Sigma} F(A) = 4 \pi N$ where $N$ is an integer called the vortex 
number.

There is a theorem by Taubes (cf.\cite{Taubes}) and Bradlow ~\cite{B}, 
 which says that the moduli space of vortices
is parametrised by the zeroes of $\Psi$. 
Thus the moduli space is the symmetric product, ${\rm Sym}^{N}(\Sigma).$ When the Riemann surface is 
$\CC P^1$, the moduli space is  complex projective space, ${\rm Sym}^{N}(\CC P^1) = \CC P^N.$

\subsection{The metric and symplectic form}

Let ${\mathcal A} $ be the space of all unitary connections on $L$ and 
$\Gamma (\Sigma, L)$ be sections of $L$.

There is a natural norm on ${\mathcal A} $ and $\Gamma (\Sigma, L)$, 
as follows:

Let $A = A_0 + \alpha$ where $\alpha $ is a unitary one form on the Riemann surface and $A_0$ is fixed.
 Then we define $||A||^2 = \int_{\Sigma} *\alpha \wedge  \alpha $
and $|| \Psi ||^2 = \int_{\Sigma} | \Psi|_H^2 d \mu $. (The Hodge star acting on an imaginary valued  $1$ form is spelt out below). 

In general, if $\tau$ is a $k$- form on a manifold $\Sigma$,  
the norm of $\tau$ is given by $ \int_{\Sigma} * \tau \wedge \tau ,$ 
$*$ being the Hodge star.

With these norms we consider the Sobolev space $H_3 ( {\mathcal A} )$ to be
$\{ A \in L^2({\mathcal A}) {\rm \;} | {\rm \;} || A ||^2 +   || F(A)||^2 < \infty,  \} $ (since $d (F(A)) = 0$ on a Riemann surface), 
where $F(A)$ is defined, for the time being, in the distributional sense.
Let $H_3 (\Gamma (\Sigma, L)  )$ to be
$\{ \Psi \in L^2(\Gamma (\Sigma, L)) {\rm \;} | {\rm \;} || \Psi ||^2  +  || \nabla \Psi||^2  +  || \nabla^2 \Psi||^2   < \infty \}$ where the covariant derivative $\nabla$ is taken w.r.t a fixed 
connection, say $A_0$,  and is defined (for the time being) 
in the distributional sense. 
 
On a Riemann surface $(A, \Psi) \in H_3$ implies (by Sobolev  lemma  ) that 
$A$ and $ \Psi$ are actually $C^1$ (~\cite{GH}, page 86).

Recall that  $H_3 (\Gamma (\Sigma, L)  )$ and $H_3 ( {\mathcal A} )$ are 
Hilbert spaces, a fact that will be useful in Marsden-Weinstein reduction.

Let ${\mathcal C} = H_3 ({\mathcal A}) \times H_3 (\Gamma (\Sigma, L))$ be the affine 
space on which 
equations $(1)$ and $(2)$ are imposed. 
Then the tangent space at $p = (A, \Psi)$ is 
$ T_p {\mathcal C} = H_3(\Omega^1(\Sigma, i \RR)) \times H_3(\Gamma(\Sigma, L))$

Once again, $(\alpha, \beta)$  in  $H_3$ implies they are actually
$C^1$ on a Riemann surface.

Let $p= (A, \Psi) \in {\mathcal C}$, $X
= ( \alpha_1, \beta)$, $Y= (\alpha_2, \eta)$
$\in T_p {\mathcal C} $  i.e.
 $\alpha_i = \alpha_i^{(0,1)} + \alpha_i^{(1,0)}$ such that 
$\overline{\alpha_i^{(0,1)}}= - \alpha_i^{(1,0)}, i = 1,2.$

On ${\mathcal C}$ one can define a metric
\begin{eqnarray*}
 {\mathcal G} ( X, Y) = \int_\Sigma *\alpha_1 \wedge \alpha_2  + i \int_\Sigma \Re < \beta, \eta>_H \omega\end{eqnarray*}
and an almost complex structure  ${\mathcal I} = \left[
\begin{array}{cc}
* & 0  \\
0 & i 
\end{array} \right] : T_p {\mathcal C} \rightarrow T_p {\mathcal C}$
where   $*: \Omega^{1} \rightarrow \Omega^{1}$ is  the Hodge star
operator on $\Sigma$ such that $*\alpha^{1,0} = - i \alpha^{1,0}$ and $*\alpha^{0,1} = i \alpha^{0,1}$ (i.e. it makes $A^{0,1}$ the holomorphic coordinate
on $ {\mathcal A}$). 

It is easy to check that ${\mathcal G} $ is positive definite. In fact, if  $\alpha_1 = \alpha^{(1,0)} + \alpha^{(0,1)} = a dz - \bar{a} d\bar{z} = i(\stackrel{\cdot}{A_1} dx + \stackrel{\cdot}{A_2 dy}) $ is an 
imaginary valued $1$-form,  $* \alpha_1 = -i (adz + \bar{a} d \bar{z})$ and
$${\mathcal G}(X, X) = \displaystyle  \int_\Sigma 4 |a|^2 dx \wedge dy  +   \int_\Sigma |\beta|^2_H h^2 dx \wedge dy$$ 
where $\omega = \frac{1}{2} h^2 dz \wedge d \bar{z} = -i h^2 dx \wedge dy$.  


We define
\begin{eqnarray*}
\Omega(X, Y) &=& -\frac{1}{2}\int_{\Sigma} \alpha_1 \wedge \alpha_2 + \frac{i}{2} \int_{\Sigma} \Re <
i \beta , \eta>_H \omega\\ 
&=& -\frac{1}{2}\int_{\Sigma} \alpha_1 \wedge \alpha_2 - \frac{1}{4}\int_{\Sigma}( \beta H \bar{\eta} - \bar{\beta} H \eta ) \omega
\end{eqnarray*}
such that $ {\mathcal G} ({\mathcal I} X, Y) = 2 \Omega ( X, Y).$ It is closed, since it is constant. In fact,  $\Omega$ is the real K\"ahler form on the affine space, ~\cite{W} page 93. It is positive, since it comes from a positive 
definite metric.

\begin{lemma}
$\Omega$ is a symplectic form on  the vortex moduli space.
\end{lemma}

\begin{proof}

Let ${\mathcal C}^{\prime} = \{ (A, \Psi) \in {\mathcal C} | \bar{\partial}_A \Psi = 0 \} $ be the submanifold of ${\mathcal C}$ defined by the second equation.

a)  We first show that  ${\mathcal C}^{\prime}$ is a symplectic submanifold with the same 
symplectic form. 
 By linearizing the equation $\bar{\partial}_A \Psi = 0 $ we get 
$\Psi \alpha_1^{0,1} = -\bar{\partial}_A \beta$
so that if 
$X= (\alpha_1, \beta) \in T_p {\mathcal C}^{\prime}$ then 
$\alpha_1^{0,1} = -\frac{\bar{\partial}_A \beta}{\Psi},$ when $\Psi \neq 0$.  

(Note that:

(i) $\bar{\partial}_{A} \beta $ takes values in $H_3(\Omega^{0,1}(\Sigma)) \otimes H_3(\Gamma (L, \Sigma))$ and hence $\frac{\bar{\partial}_{A} \beta}{ \Psi}$  is a $(0,1)$ form. 

(ii) $\alpha_1$ is $C^1$ since it is in $H_3 (\Omega^1(\Sigma, i \RR))$ and we are on 
a Riemann surface.

(iii) We show that $\Psi$ (which is $C^1$) is zero only on a set of measure zero. 
Suppose $\Psi$ is zero on a set of non-zero measure, say on an open set $U$.
Then, since  $A^{0,1} = \bar{\partial} {\rm log} \Psi$ and is not finite on $U$, $||A||^2 $ will not be finite, a contradiction.)

Suppose that, 
\begin{eqnarray*}
\Omega(X, Y) &=& -\frac{1}{2}\int_{\Sigma} \alpha_1 \wedge \alpha_2 - \frac{1}{4}\int_{\Sigma}( \beta H \bar{\eta} - \bar{\beta} H \eta ) \omega = 0
\end{eqnarray*}
for all $Y = (\alpha_2, \eta) \in T_P {\mathcal C}^{\prime}.$

Then, we show that $X = ( \alpha_1, \beta)\in T_P {\mathcal C}^{\prime} $ is zero.

Recall that $\alpha_1^{0,1} = -\frac{ \bar{\partial}_{A} \beta }{ \Psi}$, a.e. Take $\alpha_2 = * \alpha_1$ and  $\eta = i \beta$.

Note that this  satisfies the linearized equation 
$\Psi \alpha_2^{0,1} = - \bar{\partial}_{A} \eta $ and hence
$(\alpha_2, \eta) \in  T_P {\mathcal C}^{\prime}$.

From the expression of the symplectic form, we get 
$$||\alpha_1||^2   + || \beta ||^2 = 0.$$ 
Therefore, $\alpha_1 =0$ a.e. and $\beta = 0$ a.e.

Since $\alpha_1$ and $\beta$ are $C^1$, therefore, $\alpha_1$ and $\beta$ 
are zero on the complement of a set of measure zero implies that they are 
identically zero.

b) In this part we mostly follow ~\cite{D1} with slight modification. 

Let $\zeta \in \Omega(\Sigma, i{\RR })  $ be the Lie algebra of the
gauge group (the gauge group element being $g = e^{ \zeta}$ ); note that 
$\zeta$ is purely imaginary.  
It generates a vector field $X_{\zeta}$ on ${\mathcal C}^{\prime}$ as follows :
$$X_{\zeta} (A, \Psi) = (d \zeta, -\zeta \Psi) \in T_p
{\mathcal C}^{\prime} $$ where $ p = (A, \Psi) \in {\mathcal C}^{\prime}.$

(Note that $(d \zeta, -\zeta \Psi) \in T_p{\mathcal C}^{\prime} $ follows from the fact that $\bar{\partial}_A \Psi=0$ is invariant under the gauge 
transformation and can be checked easily too).  

We show next that $X_{\zeta}$ is Hamiltonian.
Let us define
$H_{\zeta}: {\mathcal C}^{\prime} \rightarrow {\CC}$ as follows:

$H_{\zeta}(A, \Psi) = \frac{1}{2} [\int_{\Sigma} \zeta \cdot (F(A) - \frac{1}{2} ( 1 - |\Psi|^2_H) \omega] $ 
to be the Hamiltonian for the gauge group action.
Then for $X = (\alpha, \eta) \in T_p {\mathcal C}^{\prime}$,
 \begin{eqnarray*}
 dH_{\zeta} ( X ) & = & \frac{1}{2} \int_{\Sigma} \zeta d \alpha  + \frac{1}{4}\int_M \zeta   
( \Psi H \bar{\eta} + \bar{\Psi} H \eta )  \omega    \\
 &= & -\frac{1}{2} \int_{\Sigma} (d \zeta) \wedge \alpha  -\frac{1}{4}  \int_{\Sigma} [(-\zeta \Psi) H \bar{\eta} - \overline{(-\zeta \Psi)} H \eta] \omega  \\ 
 & = & \Omega ( X_{\zeta},  X ),
 \end{eqnarray*}
where we use that $\bar{\zeta} = - \zeta$.

 Thus we can define the moment map $ \mu : {\mathcal C}^{\prime} \rightarrow
 \Omega^2 ( \Sigma, i{\RR} )= {\mathcal G}^* $ ( the dual of the Lie
 algebra of the gauge group)  to be $$ \mu ( A, \Psi)
 \stackrel{\cdot}{=} \frac{1}{2} [F(A) - \frac{1}{2} ( 1-|\Psi|_{H}^2)  \omega]. $$ Thus equation $(1)$ is $\mu = 0$ and the form descends as a symplectic form
to $\mu^{-1}(0)/ {\cG}$. This follows from Marsden and Weinstein, ~\cite{MW}, who proved the symplectic reduction even for infinite-dimensional case.
 \end{proof}

The moduli space of vortices, $Sym^{N} (\Sigma)$ is a smooth K\"{a}hler manifold with the Manton-Nasir K\"{a}hler form $\omega_{MN},$
(which is equivalent to our metric).

This  metric is given by ( ~\cite{MN}, eqn (2.14)).

This is the same as the descendent of our metric:
\begin{eqnarray*}
 {\mathcal G}(X, X) &=&   \int_\Sigma 4 |a|^2 dx \wedge dy  +  \int_\Sigma |\beta|^2_H h^2 dx \wedge dy \\
&=&    \int_\Sigma (\stackrel{\cdot}{A_1} \stackrel{\cdot}{A_1} + \stackrel{\cdot}{A_2} \stackrel{\cdot}{A_2})  dx \wedge dy  +  \int_\Sigma |\beta|^2_H h^2 dx \wedge dy
\end{eqnarray*}

The complex structure defined by us is exactly the usual complex structure
of the vortex moduli space i.e. on ${\rm Sym}^{N}(\Sigma)$. This can be proved 
using Ruback's argument mentioned in the appendix of  ~\cite{S},
 or in the mathscinet review of ~\cite{MN}. (Recall it takes $\alpha^{0,1}$ to 
$i \alpha^{0,1} $ and $\alpha^{1,0}$ to $-i \alpha^{1,0}$, and $\beta $ to $i \beta$ which is the case in Ruback's argument).

Now we restrict our attention to the case when the Riemann surface is $\CC P^1$
of radius $R$.  The Manton-Nasir K\"{a}hler form
on the moduli space of $N$ vortices on $\CC P^1$
of radius $R$ is given by, ~\cite{MN}, ~\cite{R},
\begin{eqnarray*}
\omega_{MN} = \frac{i}{2} \sum_{r,s=1}^N \left( \frac{4 R^2 \delta_{rs}}{(1+|z_r|^2)^2} + 
2 \frac{\partial \bar{b_s}}{\partial z_r} \right) d z_r \wedge d \bar{z_s}
\end{eqnarray*}
 This  real $2$-form
is exactly the symplectic reduction of $ \Omega$ on the affine space $\cC$ of 
$(A, \Psi) $ to the vortex moduli space, when the number of vortices is $N$ and the Riemann surface is $\CC P^1$ of radius $R$. 
Thus $\Omega =   \omega_{MN}.$  (This statement is true even for a compact Riemann surface of genus $g >0.$)

We study this form on the moduli space of $N$ vortices on the sphere of radius 
$R$. In the following we refer to Romao, ~\cite{R}, where  $\omega_{MN}$ is denoted by $\omega_{sam}.$
If $ \frac{1}{2 \pi} \Omega$ is integral, 
then $  \frac{1}{2 \pi} [\Omega]   = \ell [\omega_{FS}],$ for some $\ell \in \ZZ$
since $H^2 (\CC P^N, \ZZ)$ is generated by $[\omega_{FS}]$ (the cohomology class of the Fubini-Study K\"{a}hler form.) 

By ~\cite{R}, we can fine tune the volume of the sphere such that $\ell=1$.
Namely,  we  take  $R^2 = \frac{1}{2} + N$ and 
$\kappa = 2k$, $k \in \ZZ$. This satisfies the constraints that $\kappa$ and 
$ \kappa (R^2 -N)$ is an integer.  

Then $[\frac{1}{2 \pi} \Omega] = [\frac{1}{2 \pi} \omega_{MN}] = [ \omega_{FS}]$.
(One can also derive  the  formula $\frac{1}{2 \pi} [\omega_{MN}] = 2 (R^2 - N) [\omega_{FS}]$ from Manton and Nasir, ~\cite{MN} ((3.20), pg 600) ). 

 Note that the
Bradlow criterion for existence of solution holds for this choice of $R^2$, since area of the sphere is 
greater than $4 \pi N,$ ~\cite{R} equation (18), ~\cite{B}.

\subsection{Quillen Determinant Line bundle and the hyperplane bundle on $\CC P^N$} 

We denote the Quillen bundle ${\mathcal P} = {\rm det} (\bar{\partial}_{A})$ which is well defined on 
$\cC={\mathcal A} \times \Gamma (L)$ (over every $(A, \Psi)$ the fiber is that of ${\rm det} (\bar{\partial} + A^{0,1})$). 
In ~\cite{D}, following Biswas and Raghavendra's work on stable triples, ~\cite{BR}, we had given ${\mathcal P}$ a modified Quillen metric, namely, we multiply the Quillen metric $e^{-\zeta_A^{\prime}(0)}$ by the factor  $e^{- \frac{i}{4 \pi} \int_{M} | \Psi|^2_H \omega} ,$ where recall $\zeta_A(s)$ is the zeta-function 
corresponding to the Laplacian of the $\bar{\partial} + A^{0,1}$ operator. Here we calculate the curvature for this modified metric on the affine space since the details are not given in ~\cite{D}.
 The factor $e^{-\zeta_A^{\prime}(0)}$ contributes $\displaystyle\frac{i}{2\pi} \left( -\frac{1}{2}\int_{\Sigma} \alpha_1 \wedge \alpha_2\right)$ to the curvature, 
and the factor $e^{- \frac{i}{4 \pi} \int_{M} | \Psi|^2_H \omega}$ contributes
$\displaystyle\frac{i}{2\pi} \left(- \frac{1}{4}\int_{\Sigma}( \beta H \bar{\eta} - \bar{\beta} H \eta ) \omega\right)$  to the curvature.

\begin{lemma}
The curvature $\Omega_{det}$ of ${\mathcal P}$ with the modified Quillen metric is indeed $\frac{i}{ 2 \pi} \Omega$ on the affine space $\cC$. $\Omega$ is the real $(1,1)$ given above   which is positive.
\end{lemma}

\begin{proof}
Quillen, ~\cite{Quillen}, constructs the deteminant line bundle on the affine space ${\mathcal A}^{0,1}$. 

We consider the space of unitary connections, i.e $\overline{A^{0,1}} = - A^{1,0}$.
The complex structure on ${\mathcal C}$ defined by 

  ${\mathcal I} = \left[
\begin{array}{cc}
* & 0  \\
0 & i 
\end{array} \right] : T_p {\mathcal C} \rightarrow T_p {\mathcal C}$ makes 
$(A^{0,1}, \Psi)$ the holomorphic variables on ${\mathcal C}.$
This is because $* \alpha^{0,1} = i \alpha^{0,1}.$

Thus the holomorhic coordinate $w$ in ~\cite{Quillen} corresponds to 
 $A^{0,1}$ in our notation.

By Quillen's computation, the curvature two form is 
$\bar{\partial} \partial log || \sigma ||^2 $ where  
$ log || \sigma ||^2 = -\zeta_{A}^{\prime}(0)$ where $\zeta_{A}$ is the 
zeta-function corresponding to the Laplacian of $\bar{\partial}_{A}.$

We first note that $\bar{\partial} {\partial}  log || \sigma ||^2 =  {\partial}\bar{\partial} [- log || \sigma ||^2] $ since the order of the derivatives is switched.
The Quillen curvature form is $\bar{\partial} {\partial}  log || \sigma ||^2 =  {\partial}\bar{\partial} [- log || \sigma ||^2]  = {\partial}\bar{\partial} \zeta_{A}^{\prime}(0)$. 
Now
$$\frac{ \partial^2 \zeta_{A}^{\prime}(0)}{\partial w \partial \bar{w}}dw \wedge d \bar{w}  = \frac{i}{2 \pi} \int_{\Sigma} \overline{\delta A^{0,1}} \wedge \delta A^{0,1}, $$
that is,
\begin{eqnarray*}
(\frac{ \partial^2 \zeta^{\prime}(0)}{\partial w \partial \bar{w}} dw \wedge d \bar{w}) (\alpha_1^{0,1}, \alpha_2^{0,1})   &=& \frac{i}{2 \pi} \int_{\Sigma} \frac{1}{2}[\overline{\alpha_1^{0,1}} \wedge \alpha_2^{0,1} - \overline{\alpha_2^{0,1}} \wedge \alpha_1^{0,1}] \\
&=&\frac{i}{2 \pi} \int_{\Sigma} \frac{-1}{2}[\alpha_1 \wedge \alpha_2] 
\end{eqnarray*}
where recall we have used the fact that $\overline{\alpha_i^{0,1}} = - \alpha_i^{1,0}$  This  precisely corresponds to  $i / 2 \pi$ times the first term
in our K\"{a}hler form.

Next, the term 
 $e^{- \frac{i}{4 \pi} \int_{M} | \Psi|^2_H \omega} $ contributes to the curvature as:

\begin{eqnarray*}
- \frac{i}{4 \pi}\tau = - \frac{i}{4 \pi}\int_{\Sigma} <\delta \Psi \wedge \delta \bar{\Psi}>_H \omega
\end{eqnarray*}
where $- \frac{i}{4 \pi}\tau$ is a two-form on the affine space $\Gamma(\Sigma, L)$. 
Details of this calculation can be found in ~\cite{D}.
Here $\tau (\beta, \eta) = \frac{1}{2}\int_{\Sigma} (\beta H \bar{\eta} - \bar{\eta} H \beta) \omega$. 

Thus the contribution is $$- \frac{i}{4 \pi} \tau (\beta, \eta) = \frac{i}{2 \pi} \left[-\frac{1}{4}\int_{M} (\beta H \bar{\eta} - \bar{\eta} H \beta) \omega\right],$$
which is precisely corresponds to $i/2 \pi$ times 
the second term in our K\"{a}hler form.
\end{proof}

To define it on ${\cC }/ {\cG},$ we repeat the argument in ~\cite{D1}, ~\cite{D} with $B=0$. 

\begin{lemma} ${\mathcal P}$ descends to a well defined line bundle (denoted by the same symbol) on  ${\cC} / {\cG}$.
\end{lemma}

\begin{proof}
Let $D =  \bar{\partial} + A^{(0,1)}$ and let $D_g = g (\bar{\partial} + A^{(0,1)})g^{-1} $ (the gauge transformed operator), 
 then the gauge transformed Laplacian of $D$, namely $\Delta_g =  g \Delta g^{-1}.$ Thus there is an  isomorphism of eigenspaces, namely, 
 $s \rightarrow g s.$

Let $K^a(\Delta)$ be the direct sum of 
eigenspaces of the operator $\Delta$ of 
eigenvalues $< a$, over the open subset 
$U^a = \{ A^{(0,1)} | a \notin {\rm Spec} \Delta \}$ of the affine space 
${\cC}.$ The determinant line bundle is defined using the exact sequence
$$ 0 \rightarrow {\rm Ker} D \rightarrow K^a(\Delta) \rightarrow 
D(K^a(\Delta)) \rightarrow {\rm Coker} D \rightarrow 0$$ 
Thus 
one identifies 
$$\wedge^{{\rm top} }({\rm Ker} D)^* \otimes \wedge^{{\rm top} }
({\rm Coker} D)$$ with 
 $\wedge^{{\rm top}}(K^a(\Delta))^* \otimes \wedge^{{\rm top}} 
(D(K^a(\Delta)))$  (see ~\cite{Quillen},  for more details) and 
there is an isomorphism of the fibers as $D \rightarrow D_g$. 
Therefore one can identify 
$$ \wedge^{{\rm top}}(K^a(\Delta))^* \otimes \wedge^{{\rm top}} 
(D(K^{a}(\Delta))) \equiv
\wedge^{{\rm top}}(K^a(\Delta_g))^* \otimes \wedge^{{\rm top}} 
(D(K^{a}(\Delta_g))).$$

Let $U^a_g = g \cdot U^a$ where $U^a_g$ is the open set formed out of gauge 
transformation of $U^a, $ namely, $U^a_g = \{ A_g,  | a \notin Spec(g \Delta g^{-1}) \}.$ But $A_g \in U^a_g$  imples $a \notin Spec( \Delta)$ and thus   $A_g \in U^a$. Hence $U^a \subset U^a_g$. Similarly,  $U_g^a \subset U^a.$
Thus $U^a = U^a_g$.

On $U^a$ one defines the equivalence class of the fiber 
$$ \wedge^{{\rm top}}(K^a(\Delta))^* \otimes \wedge^{{\rm top}} 
(D(K^{a}(\Delta))) \equiv
\wedge^{{\rm top}}(K^a(\Delta_g))^* \otimes \wedge^{{\rm top}} 
(D(K^{a}(\Delta_g))).$$
\end{proof}

When the form $\frac{1}{2 \pi}\Omega$ is integral, this construction holds for 
a compact Riemann surface of genus $g$.
Thus ${\mathcal P}$  descends to the quotient of the affine space $\cC$ modulo the gauge group and further
restricts to a line bundle on the vortex moduli space. We consider this in the   case when the Riemann surface is $\CC P^1$ of radius $R$ where $R^2 = \frac{1}{2} + N$. We denote this line bundle again by 
${\mathcal P}$ and its curvature by $ \Omega_{det}= \frac{i}{2 \pi}\Omega$ such that 
 $\Omega$ coincides with the
 usual K \"{a}hler  form on the vortex moduli space, namely $ \omega_{MN}$, which is cohomologous to $ 2 \pi \omega_{FS}$. The Chern class of ${\mathcal P}$ is 
$\frac{1}{2 \pi} [\Omega]$ which is integral and $\Omega$ is positive.

The first Chern class of ${\mathcal P}$ , namely $ \frac{1}{2 \pi} [\Omega]$,  is equal to 
the first Chern-class of the hyperplane line bundle
on $\CC P^N$, namely $[\omega_{FS}]$.

Since the Picard variety of 
$\CC P^N $ is trivial, the hyperplane bundle is uniquely defined as a 
holomorphic line bundle by its first Chern class and so is the determinant 
bundle. Since the first Chern classes agree, they are the equivalent
 as holomorphic line bundles.

Thus we have the following proposition (where we denote  ${\mathcal P}= \det(\dbar)$).

\begin{proposition}\label{prop:proj}
The determinant line bundle ${\mathcal P}= \det(\dbar)$ of the family of Cauchy-Riemann operators on $\CC P^1$, as a line bundle on $\CC P^N$  is equivalent, as holomorphic line 
bundles, to the hyperplane line bundle  $\cO(1)$ on $\CC P^N$, if the radius $R$ of $\CC P^1$ satisfies $R^2 = \frac{1}{2} + N$.
\end{proposition} 

\section{Kodaira embedding theorem and application}\label{sect:kodaira}

Here we recall the Kodaira embedding theorem (cf. \cite{GH,Wells}) in a form that is suitable for our application.

\begin{theorem}[Kodaira embedding theorem]\label{thm:kodaira}
Let $M$ be a compact K\"ahler manifold with integral K\"ahler form $\omega$. Let $\cL\to M$ be the line bundle 
with connection $\nabla$ and curvature $\omega$. Then there is a positive integer $k_0$ (which we will assume is minimal)
such that for all $k\ge k_0$
the natural map
$$
\phi_k : M \hookrightarrow \PP(H^0(M, \cO(\cL^{\otimes k}))^*)
$$
is a holomorphic embedding. In particular, one has the equality of cohomology classes 
$[\phi_k^*(\Omega_k)] = k[\omega]$, where $\Omega_k$ denotes the K\"ahler form of the 
Fubini-Study metric on $\PP(H^0(M, \cO(\cL^{\otimes k}))^*)$.  In fact, $\phi_k^*({\mathcal O}(1)) \cong \cL^{\otimes k}$
as holomorphic line bundles on $M$, where ${\mathcal O}(1)$  is the hyperplane bundle $\PP(H^0(M, \cO(\cL^{\otimes k}))^*)$.
\end{theorem}

Here $H^0(M, \cO(\cL^{\otimes k}))$ denotes the finite dimensional vector space of all holomorphic sections of 
the line bundle $\cL^{\otimes k}$, which has a natural $L^2$-metric using the K\"ahler structure on $M$, which in turn induces the Fubini-Study K\"ahler form $\Omega_k$ on $\PP(H^0(M, \cO(\cL^{\otimes k}))^*)$.


For instance, for a compact Riemann surface of genus greater than 2, the minimal choice of $k_0$ is equal to $3$. 

\begin{proof}
Here we outline a few more details re the last half of the Kodaira embedding theorem above, which may not appear 
in the standard versions of the theorem in reference
books (cf. \cite{GH,Wells}). The vector space of all holomorphic sections $H^0(M, \cO(\cL^{\otimes k})$ of $\cL^{\otimes k}$ has 
an inner product induced by the K\"ahler volume form, 
$$
\langle \sigma_1, \sigma_2\rangle = \int_M h^{\otimes k} (\sigma_1(x), \sigma_2(x)) \frac{\omega^n}{n!}
$$
where $n= \dim_\CC(M)$ and $h^{\otimes k} $ is the induced hermitian inner product on $\cL^{\otimes k}$.
This in turn determines a Fubini-Study K\"ahler form $\Omega_k$ on $\PP(H^0(M, \cO(\cL^{\otimes k}))^*)$, 
which is the first Chern class of the hyperplane bundle ${\mathcal O}(1)$, satisfying $[\phi_k^*(\Omega_k)] = k[\omega]
\in H^{1,1}(M)\bigcap H^2(M;\ZZ)$. 
Therefore $\phi_k^*({\mathcal O}(1)) \cong \cL^{\otimes k}$
as holomorphic line bundles on $M$.
\end{proof}

\begin{proof}[Proof of Theorem \ref{thm:main}]
By Proposition \ref{prop:proj}, we know that Theorem \ref{thm:main} is true for complex projective space 
$\CC P^N$ for all positive integers $N$. By the Kodaira embedding theorem, Theorem \ref{thm:kodaira}, 
and with the $N = \dim H^0(M, \cO(\cL^{\otimes k_0} )) - 1$, Theorem \ref{thm:main} is true in the general case.
\end{proof}


\section{Symplectic variant and application}\label{sect:gromov}

The symplectic variant of the Kodaira embedding theorem is due to Gromov and we quote a version in 
\cite{Tischler}, Remark following Theorem B.

\begin{theorem}[Gromov's embedding theorem]\label{thm:gromov}
Let $M$ be a compact symplectic manifold with integral symplectic form $\omega$. Let $\cL\to M$ be the line bundle 
with connection $\nabla$ and curvature $\omega$. Then there is a positive integer $k_0$ (which we will assume is minimal)
such that for all $k\ge k_0$
there is a symplectic embedding
$$
\phi_k : M \hookrightarrow \CC P^k.
$$
In particular, $\phi_k^*(\Omega_k) = \omega$, where $\Omega_k$ denotes the symplectic form of the 
Fubini-Study metric on $ \CC P^k$.
\end{theorem}

Our next result is the symplectic analogue of Theorem \ref{thm:main}.

\begin{theorem}\label{thm:symplectic}
Any compact symplectic manifold $M$ with integral symplectic form $\omega$, 
parametrizes a natural smooth family of Cauchy-Riemann operators 
$\{\dbar_z: z\in M\}$ on $\CC P^1$ such that the determinant line bundle $\det(\dbar) \cong \cL$ as complex line bundles, where $\cL$ is the prequantum line bundle determined by $\omega$. 
\end{theorem}

\begin{proof}
By Proposition \ref{prop:proj}, we know that Theorem \ref{thm:symplectic} is true for complex projective space 
$\CC P^N$ for all positive integers $N$. By Gromov's embedding theorem, Theorem \ref{thm:gromov}, 
and with the minimal choice of $k_0$ as before, Theorem \ref{thm:main} is true in the general case.
\end{proof}

\section{Conjecture}\label{sect:conj}

We conjecture that equivariant versions of our results are true. More precisely,

\begin{conj*}
Let $M$ be a compact K\"ahler $G$-manifold $M$ with integral K\"ahler form $\omega$,
where $G$ is a compact Lie group. That is, $\omega$ is a $G$-invariant K\"ahler 
form on the $G$-manifold $M$ and all structures are $G$-invariant. Let $\cL$ denote the 
holomorphic $G$-line bundle determined by $\omega$. Then
$M$ parametrizes a natural holomorphic $G$-equivariant family of 
Cauchy-Riemann operators 
$\{\dbar_z: z\in M\}$
on a compact K\"ahler $G$-manifold $Z$ such that the 
Quillen determinant line bundle $\det(\dbar) \cong \cL^{\otimes k}$ as holomorphic
$G$-line bundles, for some sufficiently large $k$.
\end{conj*}

There is an analogous conjecture in the equivariant symplectic case.
If the conjectures are valid, then they would apply to coadjoint integral maximal orbits and 
thus to representation theory. It would also be relevant to the popularly known principle, ``quantization commutes with reduction''
of Guillemin and Sternberg, cf. \cite{GS}

\end{document}